\documentclass[twocolumn,reprint,aps,prd,longbibliography,nofootinbib]{revtex4-1}

\input{commands}

\appendixtrue

\ifappendix
\newcommand{\appref}[1]{Appendix~\ref{#1}}

\else
\usepackage{xr}

\makeatletter
\newcommand*{\addFileDependency}[1]{
  \typeout{(#1)}
  \@addtofilelist{#1}
  \IfFileExists{#1}{}{\typeout{No file #1.}}
}
\makeatother
 
\newcommand*{\myexternaldocument}[1]{%
    \externaldocument{#1}%
    \addFileDependency{#1.tex}%
    \addFileDependency{#1.aux}%
}
\myexternaldocument{appendix}

\newcommand{\appref}[1]{section~\ref{#1} of the SI}

\fi

\makeatother
\begin{document}
\title{Decomposing information into copying versus transformation}





 \author{Artemy Kolchinsky$^1$ and Bernat Corominas-Murtra$^{2}$}
 \thanks{Author for correspondence: bernat.corominas-murtra@ist.ac.at}

  \affiliation{
  $^1$ Santa Fe Institute, 1399 Hyde Park Road, Santa Fe, NM 87501, USA\\
  $^2$ Institute of Science and Technology Austria, Am Campus 1, A-3400, Klosterneuburg, Austria}


\begin{abstract}
In many real-world systems, information can be transmitted in two qualitatively different ways: by {\em copying} or by {\em transformation}. {\em Copying} occurs when messages are transmitted without modification, e.g.,  when  an  offspring  receives  an  unaltered  copy  of  a  gene from its parent. {\em Transformation} occurs when messages are modified systematically during transmission, e.g., when mutational biases occur during genetic replication.  Standard information-theoretic measures do not distinguish these two modes of information transfer, although they  may  reflect  different mechanisms and have different functional consequences. Starting from a few simple axioms, we derive a decomposition of mutual information into the information transmitted by copying versus the information transmitted by transformation. We begin with a  decomposition that applies when the source and destination of the channel have the same set of messages and a notion of message identity exists. We then generalize our decomposition to other kinds of channels, which can involve different source and destination sets and broader notions of similarity. In addition, we show that copy information can be interpreted as the minimal work needed by a physical copying process, which is relevant for understanding the physics of replication. We use the proposed decomposition to explore a model of  amino acid substitution rates. Our results apply to any system in which the fidelity of copying, rather than simple predictability, is of critical relevance. 

\end{abstract}

  \maketitle


  \section{Introduction}

Shannon's information theory provides a powerful set of tools for quantifying and analyzing information transmission.  A particular measure of interest is \emph{mutual information}, which is the most common way of quantifying the amount of information transmitted from a source to a destination. Mutual information has fundamental interpretations and operationalizations in a variety of domains, ranging from telecommunications~\cite{Shannon1948, cover_elements_2012},  gambling and investment~\cite{kelly_new_1956,barron_bound_1988,cover_universal_1996}, biological evolution~\cite{donaldson-matasci_fitness_2010}, statistical physics~\cite{sagawa_second_2008,parrondo_thermodynamics_2015}, and many others.  
Nonetheless, it has long been observed~\cite{pierce1961introduction, Corominas-Murtra:2013} that mutual information does not distinguish between a situation in which the destination receives  a \emph{copy} of the source message versus one in which the destination receives some systematically \emph{transformed} version of the source message (where ``systematic'' refers to transformations that do not arise purely from noise). 

As an example of where this distinction matters, consider the transmission of genetic information during biological reproduction. When this process is modeled as a communication channel from parent to offspring,  the amount of transmitted genetic information is often quantified by mutual information~\cite{bergstrom_transmission_2011,penner_sequence_2011,simonetti_mistic:_2013,butte_mutual_1999,ramani_exploiting_2003}.  During replication, however, genetic information is not only copied but can also undergo systematic transformations in the form of nonrandom mutational biases. For instance, in the DNA of most organisms, $\texttt{A}\leftrightarrow \texttt{G}$ and $\texttt{C}\leftrightarrow \texttt{T}$ mutations occur more frequently than $\texttt{A}\leftrightarrow \texttt{C}$, $\texttt{A}\leftrightarrow \texttt{T}$, $\texttt{G}\leftrightarrow \texttt{C}$, and $\texttt{G}\leftrightarrow \texttt{T}$ mutations~\cite{li1984nonrandomness,yang1994estimating,graur_fundamentals_2000}. That means that some information about parent nucleotides is preserved even when those nucleotides undergo mutations. Mutual information does not distinguish which part of genetic information is transmitted by exact copying and which part is transmitted by mutational biases. However, these two modes of information transmission are driven by different mechanisms and have dramatically different evolutionary and functional implications, given that mutations are more likely to lead to deleterious consequences.  

The goal of this paper is to find a general decomposition of the information transmitted by a channel into contributions  from copying versus from transformation.
In \cref{fig:Channel_example}, we provide a schematic that visually illustrates the problem. Essentially, we seek a decomposition of transmitted information into copy and transformation that  distinguishes the example provided in (\cref{fig:Channel_example}a), where the copy is perfect, from the one provided in (\cref{fig:Channel_example}b), where the message has been systematically scrambled, from the one provided in (\cref{fig:Channel_example}c), where the channel is completely noisy. Of course, we want also such a decomposition to apply in less extreme situations, where part of the information is copied and part is transformed.  


The distinction between copying and transformation is important in many other domains beyond the case of biological reproduction outlined above.
For example, in many   models of animal communication and language evolution,  agents exchange signals across noisy channels and then  use these signals to try to agree on common  referents in the external world~\cite{seyfarth1980vervet,Hurford:1989, Nowak:1999, Cangelosi:2002,Komarova:2004,Niyogi:2006,Steels:2003,Corominas-Murtra:2013}. In such models, successful communication occurs when information is transmitted by copying; if signals are systematically transformed --- e.g., by scrambling --- the agents will not be mutually intelligible, even though mutual information between them may be high.  
As another example, the distinction between copying and transformation may be relevant in the study of information flow during biological development, where recent work has investigated the ability of regulatory networks to  decode development signals, such as positional information, from gene expression patterns~\cite{Petkova:2019}. In this scenario, information is copied when developmental signals are decoded correctly, and transformed when they are systematically  decoded in an incorrect manner. 
Yet other examples are provided by Markov chain models, which are commonly used to study computation and other dynamical processes in physics~\cite{van1992stochastic}, biology~\cite{de2002modeling} or sociology~\cite{sorensen1978mathematical}, among other fields. In fact, a Markov chain can be seen as a communication channel in which the system state transmits information from the past into the future.  In this context, copying occurs  when the system maintains its state constant over time (remains in fixed points) and transformation occurs when the state undergoes systematic changes (e.g., performs some kind of non-trivial computations).

Interestingly, while the distinction between copy and transformation information seems natural, it has not been previously considered in the information-theoretic literature.  This may  be partly due to the different roles that information theory has historically played: on one hand, a field of applied mathematics  concerned with the engineering problem of optimizing information transmission (its original purpose); on the other, 
a set of quantitative tools for describing and analyzing intrinsic  properties of real-world systems. 
Because of its origins in engineering, much of information theory --- including Shannon's channel-coding theorem, which established mutual information as a fundamental measure of transmitted information~\cite{shannon1959coding, ash_information_2012, cover_elements_2012} ---  is formulated under the assumption of  
an external agent who can appropriately encode and decode information for transmission across a given communication channel, in this way accounting for any transformations performed by the channel.    
However, in many real-world systems, there is no additional external agent who codes for the channel~\cite{Hopfield:1994, Corominas-Murtra:2013}, and one is interested in quantifying the ability of a channel to copy information without any additional encoding or decoding. 
This latter problem is the main subject of this paper.


\begin{figure*}
\begin{center}
\includegraphics[width=14.5cm]{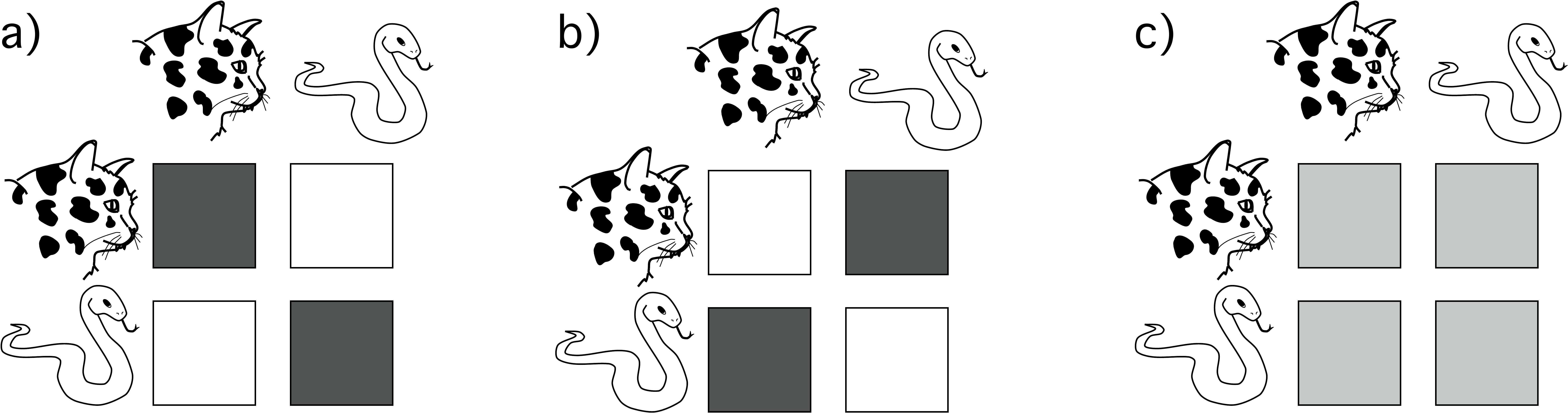}
\caption{
An illustration of the problem of copy and transformation. Consider three channels, each of which can transmit two messages, indicated by \emph{cat} and \emph{snake} (e.g., alarm calls in an animal communication system). In all panels, the rows indicate the message selected at the source, the columns indicate the message received at the destination, and the shade of the respective square indicates the conditional probability of the destination message given source message. 
For the channel in $(a)$, all information is copied: the channel maps \emph{cat}$\to$\emph{cat} and \emph{snake}$\to$\emph{snake} with probability 1. For the channel in $(b)$, all information is transformed: the channel maps \emph{cat}$\to$\emph{snake} and \emph{snake}$\to$\emph{cat} with probability 1. Note that for any source distribution, the mutual information between source and destination is the same in $(a)$ and $(b)$. The channel in $(c)$ is completely noisy: the probability of receiving a given message at the destination does not depend on the message selected at the source, and the mutual information between source and destination is $0$. Observe that transformation is different from noise, in that it still involves the transmission of information.}
\label{fig:Channel_example}
\end{center}
\end{figure*}

A final word is required to motivate our information-theoretic approach. It is standard to characterize the ability of a channel to copy messages via the ``probability of error''~\cite{cover_elements_2012}, which we indicate as $\epsilon$. In particular, $\epsilon$ is the probability that the destination receives a different message than the one that was sent by the source, while 
$1-\epsilon$ is the probability that the destination receives the same message as was sent by the source.  
However, for our purposes, this approach is insufficient.  First of all,
while $1-\epsilon$ quantifies the propensity of a channel to copy information,
$\epsilon$ does not quantify the  propensity to transmit information by transformation, since $\epsilon$  increases both in the presence of transformation and in the presence of noise (in other words, $\epsilon$ is high both in a channel like \cref{fig:Channel_example}b and a channel like \cref{fig:Channel_example}c).  Among other things, this means that $1-\epsilon$ and $\epsilon$ cannot be used to compute a channel's ``copying efficiency''  (i.e., which portion of the total  information transmitted across a channel is copied).  Second, and more fundamentally, $\epsilon$ and $1-\epsilon$ are not information-theoretic quantities, in the sense that they do not measure an \emph{amount} of information.  For instance, $1-\epsilon$ is bounded between 0 and 1 for all channels, whether considering a simple binary channel or a high-speed fiber-optic line. In the language of physics, one might say that $\epsilon$ is an intensive property, rather than an extensive one that scales with the size of the channel.   We instead seek measures which quantify the amount of copied and transformed information, and which can grow as the capacity of the channel under consideration increases.

In this paper, we present a decomposition of information that distinguishes copied from  transformed information.  We derive our decomposition by proposing four natural axioms that copy and transformation information should satisfy, and then identifying the unique measure that satisfies these axioms. 
Our resulting measure is  easy to compute and can be used to decompose either the total mutual information flowing across a channel, or the  specific mutual information corresponding to a given source message, or an even more general measure of acquired information called \emph{Bayesian surprise}.

The paper is laid out as follows. We present our approach in the next section. 
In \cref{sec:general}, we show that while our basic decomposition is defined for discrete-state channels where the source and destination share the same set of possible messages (so that the notion of ``exact copy'' is simple to define), our measures can be generalized to channels with different source and destination messages, to continuous-valued channels, and to other definitions of copying. We also discuss how our approach relates to \emph{rate-distortion} in information theory~\cite{cover_elements_2012}. 
In \cref{sec:TDcosts}, we show that our measure can be used to quantify the 
thermodynamic efficiency of physical copying processes, a central topic in biological physics. 
In \cref{sec:PAM}, we demonstrate our measures on a real world dataset of amino acid substitution rates.





\section{Copy and Transformation Information}

\label{sec:mainmeasure}

\subsection{Preliminaries}

\label{sec:Basics}

We briefly present some basic concepts from information theory that will be useful for our further developments.

We use the random variables $X$ and $Y$ to indicate the source and destination, respectively, of a communication channel (as defined in detail below). 
We assume that the source $X$ and destination $Y$ both take outcomes from the same countable set  
$\Alphabet$.  We use  $\SS$ to indicate the set of all probability distributions whose support is equal to or a subset of $\Alphabet$. 
We use notation like $\prior, \priorb, \dots \in \SS$ to indicate marginal distributions over $Y$, and $\post, \postb, \dots \in \SS$  to indicate   conditional distributions over $Y$, given the event $X=x$.  Where clear from context, we will simply write $p(y),q(y),\dots$ and $p(y\vert x), q(y\vert x), \dots$, and drop the subscripts. 

For some distribution $p$ over random variable $X$, we write the Shannon entropy as $H(p(X)):=-\sum_x p(x) \log p(x)$, or simply $H(X)$. 
For any two distributions $s$ and $q$ over the same set of outcomes, 
the {\em Kullback-Leibler} (KL) divergence is defined as
\begin{equation}
\DKL(s \Vert q):=\sum_x s(x)\log \frac{s(x)}{q(x)}\eqeol.
\label{eq:KL}
\end{equation}
KL is non-negative and equal to 0 if and only if $s(x)=q(x)$ for all $x$. It is infinite when the support of $s$ is not a subset of the support of $q$.
In this paper we will also make use of the KL between Bernoulli distributions  --- that is, distributions over two states of the type $(a, 1-a)$ --- which is sometimes called ``binary KL''.  We will use the notation $\bKL(a, b)$ to indicate the binary KL,
\begin{align}
\bKL(a, b) \coloneqq a\log\frac{a}{b} + (1-a) \log \frac{1-a}{1-b} \eqeol.
\label{eq:bKL}
\end{align}
We will in general assume that $\log$s are in base 2 (so information is measured in bits), unless otherwise noted. 

In information theory, a \emph{communication channel} specifies the conditional probability distribution of receiving different messages at a destination given  messages transmitted by a source. Let $p_{Y|X}(y\vert x)$ indicate such a conditional probability distribution. 
The amount of intrinsic noise in the channel, given some probability distribution of source messages $s_X(x)$, is the conditional Shannon entropy $H(Y|X):=-\sum_{x} s(x)\sum_y p(y|x) \log p(y|x)$. 
The amount of information transferred across a communication channel
is  quantified using the \emph{mutual information} (MI) between the source and the destination~\cite{cover_elements_2012},
\begin{align}
I_p\YX := \sum_x s(x) \sum_y p(y\vert x) \log \frac{p(y\vert x)}{p(y)} \eqeol ,
\label{eq:intromi}
\end{align}
where $p(y)$ is the marginal probability of receiving message $y$ at the destination, defined as
\begin{equation}
p(y):=\sum_x s(x) p(y\vert x)\eqeol.
\label{eq:pYy}
\end{equation}
When writing $I_p\YX$, we will omit the subscript $p$ indicating the channel where it is clear from context. 
MI is a fundamental measure of information transmission, and can be operationalized in numerous ways~\cite{cover_elements_2012}. It is non-negative, and large when (on average) the uncertainty about the message  at the destination decreases by a large amount, given the source message.
MI can also be written as a weighted sum of so-called \emph{specific MI}\footnote{The reader should be aware that the term ``specific MI'' has been used to refer to two different measures in the literature~\cite{deweese1999measure}.  The version of specific MI used here, as specified by \cref{eq:kl0}, is also sometimes called ``specific surprise''.} terms~\cite{deweese1999measure,butts_how_2003,wibral2015bits}, one for each outcome of $X$,
\begin{align}
I\YX = \sum_x s(x) I\Yx \eqeol,
\end{align}
where the specific MI for outcome $x$ is given by
\begin{align}
I\Yx & := \sum_y p(y\vert x) \log \frac{p(y\vert x)}{p(y)} = \DPP \eqeol.\label{eq:kl0}
\end{align}
Each $I\Yx$ indicates the contribution to MI arising from the particular source message $x$.
We will sometimes use the term \emph{total mutual information} (total MI) to refer to \cref{eq:intromi}, so as to distinguish it from specific MI.

Specific MI also has an important Bayesian interpretation.
Consider an agent who begins with a set of prior beliefs about $Y$, as specified by the prior distribution $\prior(y)$.  The agent then updates their beliefs conditioned on the event $X=x$, resulting in the posterior distribution $\post$.  The KL divergence between the posterior and the prior, $\DPP$ (\cref{eq:kl0}), is called \emph{Bayesian surprise}~\cite{itti_bayesian_2006}, and 
quantifies the amount of information acquired by the agent. It reaches its minimum value of zero, indicating that no information is acquired, if and only if the prior and posterior distributions match exactly. Bayesian surprise plays a fundamental role in Bayesian theory, including in the design of optimal experiments~\cite{lindley_measure_1956,stone_application_1959,bernardo_expected_1979,chaloner_bayesian_1995} and the selection of ``non-informative priors''~\cite{bernardo_reference_1979,berger_formal_2009}. Specific MI is a special case of Bayesian surprise, when the prior $p_Y$ is the marginal distribution at the destination, as determined by a choice of source distribution $s_X$ and channel $p_{Y|X}$ according to \cref{eq:pYy}.
In general, however, Bayesian surprise may be defined for any desired prior $p_Y$ and posterior distribution $p_{Y\vert x}$, without necessarily making reference to a source distribution $s_X$ and communication channel $p_{Y|X}$. 

Because Bayesian surprise is a general measure that includes specific MI as a special case, we will formulate our analysis of copy and transformation information in terms of Bayesian surprise, $\DPP$. 
Note that while the notation $\post$ implies conditioning on the event $X=x$, formally $\post$ can be any distribution whatsoever. 
Thus, we do not technically require that there exist some full joint or conditional probability distribution over $X$ and $Y$. Throughout the paper we will refer to the distributions $\post$ and $\prior$ as the ``posterior'' and ``prior''. 

Proofs and derivations are contained in the 
\ifappendix
appendices\else supplementary information (SI)\fi.



\subsection{Axioms for copy information}
\label{sec:AxiomsDI}

We propose that any measure of copy information should satisfy a set of four axioms.
Our setup is motivated in the following way.  
First, our decomposition should apply at the level of individual source message, i.e., we wish to be able to decompose each specific mutual information term (or more generally, Bayesian surprise) into a non-negative \emph{(specific) copy information} term and a non-negative  \emph{(specific) transformation information} term.  Second,  we postulate  that if there are two channels with the same marginal distribution at the destination, then the channel with the larger $p_{Y|X}(x|x)$ (probability of destination getting message $x$ when the source transmits message $x$) should have larger copy information for source message $x$ (this is, so to speak, our ``central axiom''). This postulate can also be interpreted in a Bayesian way. Imagine two Bayesian agents with the same prior distribution over beliefs, $p_Y$, who update their beliefs conditioned on the event $X=x$. We postulate that the agent with the larger posterior probability on $Y=x$ should have greater copy information.

Formally, we assume that each copy information term is a real-valued function 
of the posterior distribution, the prior distribution, and the source message $x$, written generically as $F(\post, \prior, x)$.
Given any measure of copy information $F$, the transformation information associated with message $x$ is then the remainder of $\DPP$  beyond $F$,
\begin{align}
\Ftrans(\post, \prior, x) := \DPP - F(\post, \prior, x) \eqeol .
\label{eq:ftrans}
\end{align}
We now propose a set of axioms that any measure of copy information $F$ should satisfy.

First, we postulate that copy information should be bounded between 0 and the Bayesian surprise, $\DPP$. Given \cref{eq:ftrans}, this guarantees that both $F$ and $\Ftrans$ are non-negative.
\begin{axiom} 
$F(\post, \prior, x) \ge 0$.
\label{axiom:Nonnegative}
\end{axiom}
\begin{axiom} 
$F(\post, \prior, x) \le \DPP$.
\label{axiom:UpperBound}
\end{axiom}
\noindent
Then, we postulate that copy information for source message $x$ should increase monotonically as the posterior probability of $x$ increases, assuming the prior distribution is held fixed (this is the ``central axiom'' mentioned above).
\begin{axiom} 
If $\postx \le \postbx$, then $F(\post, \prior, x) \le F(\postb, \prior,x)$.
\label{axiom:Monotonicity}
\end{axiom}
\noindent
In \appref{app:proofs}, we show that any measure of copy information that satisfies the above three axioms must obey  $F(\post,\prior,x)=0$ whenever $\postx \le \priorx$.  We also show that one particular measure of copy information, which is called $\Dcopy$ and is discussed in the next section, is the largest measure that satisfies the above three axioms.  However, the three axioms do not uniquely determine what happens when  $\postx  > \priorx$. This means that $\Dcopy$ is not unique, and in fact there are some trivial measures (such as  $F(\post,\prior,x)=0$ for all $\post$, $\prior$, and $x$) that also satisfy the above axioms.
Such trivial cases are excluded by our final axiom, which  states that for all prior distributions and all posterior probabilities $\postx  > \priorx$, 
there are posterior distributions that contain 
\emph{only} copy information.  As we'll see below, $\Dcopy$ is the unique satisfying measure once this axiom is added.
\begin{axiom} 
\label{axiom:purelycorrect}
For any $\prior$ and $c \in [\priorx, 1]$, there exists  a posterior distribution $\post$ such that $\postx = c$ and $F(\post, \prior,x)=\DPP$.
\end{axiom}


\subsection{The measure $\Dcopy$}

\label{sec:measures}

We now present $\Dcopy$, the unique measure 
that satisfies the four copy information axioms proposed in the last section. 
Given a prior distribution $\prior$,  posterior distribution $\post$, and source message $x$,  this measure  is defined as
\begin{multline}
\DcopyPP = \\
\begin{cases}\bKL(\postx , \priorx) & \text{if $\postx > \priorx$}\\
0 & \text{otherwise}
\end{cases}
\eqeol,\label{eq:DI2}
\end{multline}
where we have used the notation of \cref{eq:bKL}. 
We now state the main result of our paper:
\begin{thm}
$\Dcopy$ is the unique measure  
which satisfies \cref{axiom:UpperBound,axiom:Monotonicity,axiom:purelycorrect,axiom:Nonnegative}. 
\label{thm:unique}
\end{thm}
\noindent
In the \appref{app:satisfies} we demonstrate that $\Dcopy$ satisfies all the axioms, and in the \appref{app:proofs} we prove that it is the only measure that satisfies them. We further
show that if one drops \cref{axiom:purelycorrect}, then $\Dcopy$ is the largest possible measure that can satisfy the remaining axioms. 

Given the definition of $\Ftrans$ in \cref{eq:ftrans}, $\Dcopy$ also defines a non-negative measure of transformation information, which  we call $\Dtrans$,
\begin{align*}
\DtransPP = \DPP - \DcopyPP \eqeol .
\end{align*}

\subsection{Decomposing mutual information}
\label{sec:midecomp}

We now show that $\Dcopy$ and $\Dtrans$ allow for a decomposition of mutual information (MI) into  {\em MI due to copying} and {\em MI due to transformation}.   
Recall that MI can be written as an expectation over specific MI terms, as shown in \cref{eq:kl0}. Each specific MI term can be seen as a Bayesian surprise, where the prior distribution is  the marginal distribution at the destination (see \cref{eq:pYy}), and the posterior distribution is the conditional distribution of destination messages given a particular source message. Thus, our definitions of $\Dcopy$ and $\Dtrans$ provide a non-negative decomposition of each specific MI term,
\begin{align}
I_p\Yx = \DcopyPP + \DtransPP \eqeol.
\label{eq:Ip2}
\end{align}
In consequence, they also provide a non-negative
decomposition of the total MI into two non-negative terms: the \emph{total copy information} and the \emph{total transformation information}, 
\begin{align*}
I_p\YX = \Icopy_p\XrY + \Itransform_p\XrY \eqeol,
\end{align*}
where $\Icopy_p\XrY$  and  $\Itransform_p\XrY$ are given by
\begin{align}
\Icopy_p\XrY & := \sum_x s(x) \DcopyPP \eqeol, \label{eq:Icopytotal}\\
\Itransform_p\XrY & := \sum_x s(x) \DtransPP \eqeol. \label{eq:Itransformtotal}
\end{align}
(When writing $\Icopy$ and $\Itransform$, we will often omit the subscript $p$ where the channel is clear from context.) 
By a simple manipulation, we can also decompose the marginal entropy of the destination $H(Y)$ into three non-negative components:
\begin{align}
H(Y) = \Icopy\XrY + \Itransform\XrY + H({Y|X}) \eqeol .
\label{eq:threeway}
\end{align}
Thus, given a channel from $X$ to $Y$, the uncertainty in $Y$ can be written as the sum of the copy information from $X$, the transformed information from $X$, and the intrinsic noise in that channel from $X$ to $Y$.

For illustration purposes, we plot the behavior of $\Icopy$ and $\Itransform$ in the classical binary symmetric channel (BSC) in \cref{Fig:Binary_Ch}  (see caption for details). More detailed analysis of copy and transformation information in the BSC is discussed in \appref{app:bsc}.

It is worthwhile to point several important differences between our proposed measures and MI.

First, in the definitions of $\Icopy\XrY$ and $\Itransform\XrY$, the notation $\XrYnp$ indicates that 
$X$ is the source and $Y$ is the destination.  This is necessary because, unlike MI, $\Icopy$ and $\Itransform$ are in general non-symmetric, so it is possible that $\Icopy\XrY \ne \Icopy\YrX$, and similarly for $\Itransform$. We also note that the above form of $\Icopy$ and $\Itransform$, where they are written as sums over individual source message, is sometimes referred to as {\em trace-like} form in the literature, and is a commonly desired characteristic of information-theoretic functionals~\cite{Hanel:2011, Thurner:2017}.

Second, $\Icopy$ and $\Itransform$ do not obey the data processing inequality~\cite{cover_elements_2012}, and can either decrease or increase as the destination undergoes further operations.  In this respect, they are different from MI (the sum of $\Icopy$ and $\Itransform$). 
As an example, consider the case where channel $p_{Y|X}$ first transforms source message $X$ into an encrypted message $Y$, and then another channel $p_{X'|Y}$ decrypts $Y$ back into a copy of $X$ (so $X'=X$). In this example, $\Icopy(\gXrYnp{X}{X'}) > \Icopy\XrY$ even though the Markov condition $X - Y - X'$ holds.

Finally, unlike MI, $\Icopy$ and $\Itransform$ are generally non-additive when multiple independent channels are concatenated. As an  example, imagine that the source messages are bit strings of length $n$, which are transmitted through a product of $n$ independent channels,  $p(y|x) = \prod_i p_i(y_i|x_i)$.  If the source bits are independent, $s(x)=\prod_i s_i(x_i)$, it is straightforward to show that the MI between $X$ and $Y$ has the additive form $I\YX = \sum_i I(Y_i : X_i)$.  However, $\Icopy$ will generally not have this additive form, because copy information is defined in terms of the probability of exactly copying the entire source message (e.g., the entire $n$-bit long string).  Imagine that in the above example, one of the bit-wise channels carries out a bit flip, $p_i(x_i|y_i) = 1-\delta(x_i,y_i)$. In that case, the probability of receiving an exact copy of the source message at the destination is zero, and therefore $\Icopy$ is also zero regardless of the nature of the other bit-wise channels $p_j$ for $j\ne i$. If desired, it is possible to derive an additive version of $\Icopy$ by generalizing our measure with an appropriate ``loss function'', as discussed in more detail in \cref{sec:general} and \appref{app:general-vector}.

\definecolor{myorange}{rgb}{1,0.9,0.8}
\definecolor{myblue}{rgb}{0.8,0.85,1}

\begin{figure}
\begin{center}
\begin{tikzpicture}
\begin{axis}[xmin=0,xmax=1,ymin=-0.0,ymax=1, domain=0:1, samples=50,
xlabel={\textsf{error probability $\epsilon$}},
ylabel={\textsf{bits}},
y label style={at={(0.1,0.5)}},
ytick={0,0.5,1},
xtick={0,0.5,1},
legend style={at={(0.5,0.97)},
anchor=north},
height=5cm,
]
   \path[name path=axis] (axis cs:0,0) -- (axis cs:1,0);
  \addplot[color=black, line width=1pt] (x,{1+x*log2(x)+(1-x)*log2(1-x)});
  \addplot[name path=icopy, line width=0pt, forget plot]  (x,ifthenelse(x<0.5,{1+x*log2(x)+(1-x)*log2(1-x)}, 0);
  \addplot [
    fill=myblue, 
    fill opacity=1 
  ]
  fill between[
     of=icopy and axis,
     soft clip={domain=0:1},
  ];
 
  \addplot[name path=itrans, line width=0pt, forget plot] (x,ifthenelse(x>0.5,{1+x*log2(x)+(1-x)*log2(1-x)}, 0);

  \addplot [
    fill=myorange, 
    fill opacity=1,
  ]
  fill between[
     of=itrans and axis,
     soft clip={domain=0:1},
  ];  
  \legend{
  $I\YX$,
  $\Icopy\XrY$,
  $\Itransform\XrY$
  }
\end{axis}
\end{tikzpicture}
\caption{The Binary Symmetric Channel (BSC) with a uniform source distribution. We plot values of the MI $I\YX$, copy information $\Icopy\XrY$ (\cref{eq:Icopytotal}) and transformation information $\Itransform\XrY$ (\cref{eq:Itransformtotal}) for the BSC along the whole range of error probabilities $\epsilon\in[0,1]$. When $\epsilon \le 1/2$, all mutual information is $\Icopy$ (blue shading), when $\epsilon \ge 1/2$, all mutual information is $\Itransform$ (orange shading).}
\label{Fig:Binary_Ch}
\end{center}
\end{figure}

\subsection{Copying efficiency}
\label{sec:efficiency}


Our approach  provides a way to quantify which portion of the information transmitted across a channel is due to copying rather than transformation,  which we refer to as ``copying efficiency''.  Copying efficiency is defined at the level of individual source messages as 
\begin{align}
\eta_p(x) := \frac{\DcopyPP}{\DPP} \in [0,1] \eqeol,
\label{eq:eff0}
\end{align}
where the bounds come directly from \cref{axiom:Nonnegative,axiom:UpperBound}.  It can also be defined at the level of a channel as whole as
\begin{align}
\eta_p := \frac{\Icopy\XrY}{I\YX} \in [0,1] \eqeol.
\label{eq:eff1}
\end{align}
The bounds follow simply given the above results. 

For \cref{eq:eff0} and \cref{eq:eff1} to be useful efficiency measures, there should  exist channels which are either ``completely inefficient'' (have efficiency 0) or ``maximally efficient'' (achieve efficiency 1).  For the case of \cref{eq:eff0}, 
the bounds can be saturated because of \cref{axiom:purelycorrect}, which guarantees that for any source message $x$, prior $\prior$, and desired posterior probability $\postx \ge \priorx$, there exists a posterior $\post$ such that the Bayesian surprise $\DPP$ is 
composed entirely of copy information (for example, see \cref{eq:alphapost}).

One can show that the bounds in \cref{eq:eff1} 
can also be saturated.  
First, it can be verified that completely inefficient channels exist, since any channel which has $\postx \le \priorx$ for all $x \in \Alphabet$ will have $\Icopy\XrY = 0$  (note that such channels exist at all levels of mutual information).  We also show that maximally efficient channels exist, using the following result which is proved in \appref{app:purech}. 

\begin{restatable}{prop}{corpure}
\label{prop:dmi}
For any source distribution $s_X$ with $H(X) < \infty$, there exist channels $\post$ for all levels of  mutual information $I_p\YX \in[0,H(X)]$ such that $\Icopy_p\XrY = I_p\YX$.
\end{restatable}

\cref{prop:dmi} shows that it is possible to achieve all values of total copy information, which is defined at the level of a channel.  Note that this proposition does not follow immediately from \cref{axiom:purelycorrect}, which is a statement about copy information at the level of a prior $\prior$ and posterior $\post$, where  no particular relationship between $\prior$ and $\post$ is assumed.



\section{Generalization and relation to rate-distortion}
\label{sec:general}

We now show that $\Dcopy$ can be written as a particular element among a broad family of copy information measures, which generalize the formal definition of what is meant by  ``copying''.

As we showed above, $\Dcopy$ is the unique measure that satisfies the four axioms proposed in \cref{sec:AxiomsDI}.  In particular, it satisfies \cref{axiom:Monotonicity}, which states that given the same prior $\prior$, copy information should be larger for $\postb$ than $\post$ whenever $\postbx \ge \postx$.  It also satisfies \cref{axiom:purelycorrect}, which states that there exist posterior distributions that have only copy information for all possible $\postx \in [\priorx, 1]$.

These axioms are based on one particular definition of copying, which states that copying occurs when the source and destination messages match perfectly.  In fact, this can be generalized to other definitions of copying and transformation by using a \emph{loss function} $\ell(x,y)$, which quantifies the dissimilarity between source message $x$ and destination message $y$.  For a given loss function, $\ell(x,y) = 0$ indicates that  $x$ and $y$ should be considered a perfect copy of each other, while $\ell(x,y) > 0$ indicates that $x$ and $y$ should be considered as somewhat different.  Importantly,  $\ell(x,y)$ can quantify similarity in a graded manner, so that $\ell(x,y') > \ell(x,y)$ indicates that $y$ is closer to being a copy of $x$ than $y'$ (even though neither $y$ nor $y'$ may be a perfect copy of $x$).

Given an externally-specified loss function $\ell(x,y)$, one can define  \cref{axiom:Monotonicity} and \cref{axiom:purelycorrect} in a generalized manner. The generalized version of \cref{axiom:Monotonicity} states that posterior distribution $\postb$ should have higher copy information than $\post$ whenever its expected loss is lower:
\begin{manualax}{\ref*{axiom:Monotonicity}$^\mathbf{*}$}
\label{ax:mono2}
If $\mathbb{E}_{\post}[\ell(x,Y)] \ge \mathbb{E}_{\postb}[\ell(x,Y)]$, then $F(\post, \prior, x) \le F(\postb, \prior,x)$.
\end{manualax}
\noindent The generalized version of \cref{axiom:purelycorrect} states that at all values of the expected loss which are lower than the expected loss achieved by $\prior$,
 there are channels which transmit information only by copying.
\begin{manualax}{\ref*{axiom:purelycorrect}$^\mathbf{*}$}
\label{ax:purely2}
For any $\prior$ and $c \in [\min_y \ell(x,y), \mathbb{E}_{\prior}[\ell(x,Y)]]$, there exists  a posterior distribution $\post$ such that $\mathbb{E}_{\post}[\ell(x,Y)] = c$ and $F(\post, \prior,x)=\DPP$.
\end{manualax}
\noindent Note that in defining \cref{ax:purely2}, we used that $\min_y \ell(x,y)$ is the lowest expected loss that can be achieved by any posterior distribution.

Each particular loss function induces its own measure of copy information. 
In fact, as we show in \appref{app:general-ax}, there is a unique measure of copy information which satisfies \cref{axiom:Nonnegative} and \cref{axiom:UpperBound}, as defined in \cref{sec:AxiomsDI}, plus the generalized axioms \cref{ax:mono2} and \cref{ax:purely2}, as defined here in terms of the loss function $\ell(x,y)$. This generalized measure of copy information has the following form: 
\begin{align}
\label{eq:genloss}
&\DcopyGenPP := \min_{r_Y} \DKL(r_Y \Vert \prior)\\
&\qquad\qquad\quad  \text{s.t.} \quad 
\mathbb{E}_{r_Y}[\ell(x,Y)] \le \mathbb{E}_{\post}[\ell(x,Y)] . \label{eq:genlosscons}
\end{align}
Recall that the KL divergence $\DKL(r_Y \Vert \prior)$ reflects the amount of information acquired by an agent in going from prior distribution $\prior$ to posterior distribution $r_Y$. Thus, 
$\DcopyGenPP$ quantifies the minimum information that must be acquired by an agent in order to match the copying performance of the actual posterior $\post$, as measured by the expected loss. 


\cref{eq:genloss} is an instance of a ``minimum cross-entropy'' problem, which 
is closely related to the ``maximum entropy'' principle~\cite{kullback1959information,kapur1992entropy,shore1981properties}. The distribution that optimizes \cref{eq:genloss} can be written in a simple form~\cite[pp.299-300]{rubinstein_simulation_2016},
\begin{align}
w(y) = {\frac{1}{Z(\lambda)}} p_Y(y) e^{-\lambda \ell(x,y)} \nonumber 
\end{align}
where $\lambda \ge 0$ is a Lagrange multiplier chosen so that the constraint in \cref{eq:genloss} is satisfied, and $Z(\lambda)=\sum_y p_Y(y) e^{-\lambda \ell(x,y)}$ is a normalization constant. Note that whenever $\mathbb{E}_{\post}[\ell(x,Y)] \ge \mathbb{E}_{\prior}[\ell(x,Y)]$, $\lambda =0$ and $w_Y = p_Y$~\cite{rubinstein_simulation_2016}. Otherwise, $\lambda > 0$ and the constraint in \cref{eq:genloss} will be tight up to equality.  In practice, \cref{eq:genloss} can be solved by sweeping across the 1-dimensional space of possible $\lambda \ge 0$ values (it can also be solved by standard convex optimization techniques).  Once $\lambda$ is determined, the value of copy information is given by
\begin{align*}
\DcopyGen = -\lambda \mathbb{E}_{\post}[\ell(x,Y)] - \log Z(\lambda) .
\end{align*}

It can be verified that $\Dcopy$, the measure derived above, corresponds to the special case $\ell(x,y) := 1-\delta(x,y)$, which is called ``0-1 loss'' in statistics~\cite{friedman2001elements} and ``Hamming distortion'' in information theory~\cite{cover_elements_2012}
 (see \appref{app:derivation}).  

The generalized measure $\DcopyGen$ has many similarities with $\Dcopy$. 
Like $\Dcopy$, it naturally leads to a non-negative  measure of generalized transformation information,
\begin{align}
\DtransGenPP = \DPP - \DcopyGenPP \,.
\label{eq:dtransgen}
\end{align}
$\DcopyGen$ can also be used to decompose total mutual information into (generalized) total copy and transformation information, akin to \cref{eq:Icopytotal} and \cref{eq:Itransformtotal}. Finally, one can use $\DcopyGen$ to define a generalized measure of copying efficiency, following the approach described in \cref{sec:efficiency}.

While we believe $\Dcopy$, as defined via the 0-1 loss function, is a simple and reasonable choice in a variety of applications, in some cases it may also be useful to consider other loss functions.  One important example is when the source and destination have different sets of outcomes. Recall that  $\Dcopy$ assumes that the source and destination share the same set of possible outcomes, $\Alphabet$. When this assumption does not hold, generalized measures of copy and transformation information can still be defined, as long as an appropriate loss function $\ell : \mathcal{X} \times \mathcal{Y} \to \mathbb{R}$ is provided (where $\mathcal{X}$ and $\mathcal{Y}$ indicate the outcomes of the source and destination, respectively).

Another important use case occurs when the loss function specifies continuously-varying degrees of {functional} similarity between source and destination messages.  For example, imagine that  $p_{Y|X}$ is an image compression algorithm which maps raw images $X$ to compressed outputs $Y$. Research in computer vision has developed sophisticated loss functions for image compression which correlate strongly with human perceptual judgments~\cite{wang2006modern}. By defining copy information in terms of such a loss function, one could measure how much perceptual information is copied by a particular image compression algorithm.

Our generalized approach can also be used to define copy and transformation information for random variables with continuous-valued outcomes.  The 0-1 loss function, as used in $\Dcopy$, is not very  meaningful for continuous-valued outcomes, since it depends on a measure-0 property of $\post$. 
A more natural measure of copy information is produced by the squared-error loss function $\ell(x,y):=(x-y)^2$, giving 
\begin{align*}
\min_{r_Y} \DKL(r_Y \Vert \prior) \;\;\; \text{s.t.}  \;\;\; \mathbb{E}_{r_Y}[(Y-x)^2] \le \mathbb{E}_{\post}[(Y-x)^2] .
\end{align*}
This particularly optimization problem has been investigated in the maximum entropy literature, and has been shown to be  particularly tractable when $p_Y$ belongs to an exponential family~\cite{altun2006unifying,dudik2006maximum,koyejo2013representation}.




Finally, it is also possible to generalize this approach to vector-valued loss functions $\ell : \mathcal{X} \times \mathcal{Y} \to \mathbb{R}^n$, which allow one to specify dissimilarity in a multi-dimensional way.  We discuss the relevant axioms and resulting copy information measure for vector-valued loss functions in \appref{app:general-vector}. We also demonstrate that vector-valued loss functions can be used to define measures of copy and transformation information that are additive for independent channels, in the sense discussed in \cref{sec:midecomp}.

After what we discussed so far, it is natural to briefly review the similarities between our generalized approach and \emph{rate-distortion theory}~\cite{cover_elements_2012}. In rate-distortion theory, one is given a distribution over source messages $s_X$ and a ``distortion function'' $\ell(x,y)$ which specifies the loss incurred when source message $x$ is encoded with destination message $y$. The problem is to find the channel $r_{Y|X}$ which minimizes mutual information without exceeding some  constraint on the expected distortion,
\begin{align}
\min_{r_{Y|X}} \DKL(r_{Y|X}\Vert r_Y) \quad \text{s.t.} \quad \mathbb{E}_r[\ell(X,Y)] \le \alpha \,,
\label{eq:rd}
\end{align}
where $\alpha$ is an externally-determined parameter. 
The prototypical application of rate-distortion is compression, i.e., to find a compression channel $r_{Y|X}$ that has both low mutual information  and low expected distortion.  
As can be seen by comparing \cref{eq:genloss} and \cref{eq:rd}, the optimization problem considered in our definition of generalized copy information and the optimization found in rate-distortion are quite similar: they both involve minimizing a KL divergence subject to an expected loss constraint.   
Nonetheless, there are some important differences. 
First and foremost, the goals of the two approaches are different. In our approach, the aim is to decompose the information transmitted by a fixed externally-specified channel into copy and transformation. In rate-distortion, there is no externally-specified channel and the aim is instead to find an optimal channel \emph{de novo}. 
Second, our approach is motivated by a set of axioms  which postulate how a measure of copy information should behave, rather than from channel-coding considerations which are used to derive the optimization problem in rate-distortion~\cite{cover_elements_2012}.  Lastly, copy information is defined in a point-wise manner for each source message $x$, rather than for an entire set of source messages at once, as it is rate-distortion. 

We finish by noting that one can also define \cref{eq:genloss} in a channel-wise manner (by minimizing $\DKL(r_{Y|X}\Vert r_Y)$, as in \cref{eq:rd}) rather than a pointwise manner (minimize $\DKL(r_{Y|X=x}\Vert p_Y)$, as in \cref{eq:genloss}). Under that formulation, one could no longer decompose specific MI into non-negative copy and information terms, though total MI could still be decomposed in that way. Interestingly, this alternative formulation would become equivalent to the so-called \emph{minimum information principle}, a previous proposal for quantifying how much information about source messages is carried by  different properties of destination messages~\cite{globerson2009minimum}.

\section{Thermodynamic costs of copying}
\label{sec:TDcosts}
\def\kB{k}
\def\kBT{\kB T}

Given the close connection between information theory and statistical physics, many information-theoretic quantities can be interpreted in thermodynamic terms~\cite{parrondo_thermodynamics_2015}. As we show here, this includes our proposed measure of copy information, $\Dcopy$. 
Specifically, we will show that $\Dcopy$ reflects the minimal amount of thermodynamic work necessary to copy a physical entity such as a polymer molecule. This latter example emphasizes the difference between information transfer by copying versus by transformation in a fundamental, biologically-inspired physical setup. 

Consider a physical system coupled to a heat bath at temperature $T$, and which is initially in an equilibrium distribution $\pi(i) \propto e^{-E(i)/(\kBT)}$ with respect to some Hamiltonian $E$ ($\kB$ is Boltzmann's constant).  Now imagine that the system is driven to some non-equilibrium distribution $p$ by a physical process, and that by the end of the process the Hamiltonian is again equal to $E$. The minimal amount of work required by any such process is related to the KL divergence between $p$ and $\pi$~\cite{esposito2010three},
\begin{align}
\label{eq:klwork}
W \ge 
\kBT \; \DKL(p \Vert \pi) \eqeol.
\end{align}
The limit is achieved by thermodynamically-reversible processes. 
(In this subsection, in accordance with the convention in physics, we assume that all logarithms are in base $e$, so information is measured in nats.)

Recent work has analyzed the fundamental thermodynamic constraints  on copying in a physical system, for example for an information-carrying polymer like DNA~\cite{ouldridge_fundamental_2017,poulton_non-equilibrium_2018}.  Here we will generally follow the model described in \cite{ouldridge_fundamental_2017}, while using our notation and omitting some details that are irrelevant for our purposes (such as the microstate/macrostate distinction). In this model, the source $X$ represents the state of the original system (e.g., the polymer to be copied), and the destination $Y$ represents the state of the replicate (e.g., the polymer produced by the copying mechanism). 
We make several assumptions. First, the source $X$ is not modified during the copying process. Second, $X$ and $Y$ have the same Hamiltonian before and after the copying process.  Finally, we follow~\cite{ouldridge_fundamental_2017} in assuming that $Y$ is a \emph{persistent} copy of $X$, meaning that before and after the copying process, $Y$ is physically separated from $X$ and there is no interaction energy between them.  This does not preclude $X$ and $Y$ from coming into contact and interacting energetically during intermediate stages of the copying process (for instance by template binding). The assumption of persistent copying means that there are no unaccounted energetic costs involved in preparing the copying system and transporting the produced replicate (e.g., moving the replicate $Y$ to a daughter cell). 

Assume that $Y$ starts in the equilibrium distribution, indicated as $\pi_Y$ (note that by our persistent copy assumption, the equilibrium distribution cannot depend on the state of $X$).  Let $\postx$ indicate the conditional distribution of replicates  after the end of the copying process, where $x$ is the state of the original system $X$. 
Following \cref{eq:klwork}, the minimal work required to 
bring $Y$ out of equilibrium and produce replicates according to $\postx$ is given by
\begin{align}
\label{eq:minworkx}
W(x) \ge \kBT\; \DKL( \post \Vert \pi_Y ) \eqeol.
\end{align}

Note that \cref{eq:minworkx} specifies the minimal work required to create the overall distribution $\post$.  However, in many real-world scenarios, likely including DNA copying, the primary goal is to create exact copies of the original state, not  transformed versions it (such as nonrandom mutations). 
That means that for a given source state $x$, the quality of the replication process can be quantified by the probability of making an exact copy, $\postx$.  We can now ask: what is the {minimal} work required by a physical replication process whose probability of making exact  copies is at least as large as  $\postx$? To make the comparison fair, we require that the process begin and end with the same equilibrium distribution, $\pi_Y$. 
The answer is given by the minimum of the RHS of \cref{eq:minworkx} under a constraint on the exact-copy yield, which is exactly proportional to $\Dcopy$: 
\begin{align}
W^\text{exact}_\text{min}(x) & = \kBT \; \left[ \min_{r_Y : r_Y(x) \ge \postx} \DKL( r_Y \Vert \pi_Y )\right] \label{eq:minKLwork0}
 \\
& = \kBT \; \Dcopy(\post \Vert \pi_Y) \eqeol,
\label{eq:minKLwork}
\end{align}
where \cref{eq:minKLwork} follows from \appref{app:derivation}. 
The additional work that is expended by the replication process is then lower bounded by a quantity proportional $\Dtrans$,
\begin{align}
W(x) - W^\text{exact}_\text{min}(x) \ge \kBT  \Dtrans(\post\Vert \pi_Y) \,. \label{eq:excess1}
\end{align}
This shows formally the intuitive idea that transformation information contributes to thermodynamic costs but not to the accuracy of correct copying.

In most cases, a replication system is designed for copying not just one source state $x$, but an entire ensemble of source states (for example, the DNA replication system can copy a huge ensemble of source DNA sequences, not just one). Assume that $X$ is distributed according to some $s_X(x)$.  Across this ensemble of source states, the minimal amount of \emph{expected} thermodynamic work required to produce replicates according to conditional distribution $p_{Y|X}$ is given by
\begin{align}
\langle W\rangle & \ge   \kBT \; \sum_x s(x) \DKL(p_{Y|x} \Vert \pi_Y) \\
& = \kBT \left[ I_p\YX + \DKL(p_Y\Vert \pi_Y) \right] \eqeol.
\end{align}
Since KL is non-negative, the minimum expected work is lowest when the equilibrium distribution $\pi_Y$ matches the marginal distribution of replicates, $p_Y(y) = \sum_x s(x) p(y|x)$. 
Using similar arguments as above, we can ask about the minimum expected work 
required to produce replicates, assuming each source state $x$ achieves an exact-copy yield of at least $\postx$.  This turns out to be the expectation of \cref{eq:minKLwork},
\begin{align}
\langle W^\text{exact}_\text{min} \rangle &=   \kBT \; \sum_x s(x) \Dcopy(\postx \Vert \pi_Y) \\
& =  \kBT \left[\Icopy_p\XrY + \DKL(p_Y\Vert \pi_Y) \right] \eqeol.
\end{align}
The additional expected work that is needed by the replication process, above and beyond an optimal process that achieves the same exact-copy yield, is lower bounded by the transformation information,
\begin{align}
\langle W \rangle - \langle W^\text{exact}_\text{min} \rangle \ge \kBT \Itransform_p\XrY \,.
\end{align}
When the equilibrium distribution $\pi_Y$ matches the marginal distribution $p_Y$, $\langle W^\text{exact}_\text{min} \rangle$ is exactly equal $\kBT \Icopy$.  Furthermore, in this special case the thermodynamic efficiency of exact copying, defined as the ratio of minimal work to actual work, becomes equal to the information-theoretic copying efficiency of $p$, as defined in \cref{eq:eff1}:
\begin{align}
\frac{\langle W^\text{exact}_\text{min} \rangle}{\langle W \rangle} = \frac{\Icopy_p\XrY }{ I_p\YX} = \eta_p \eqeol.
\end{align}
 
As can be seen, standard information-theoretic measures, such as \cref{eq:minworkx}, bound the minimal thermodynamic costs of transferring information from one physical system to another, whether that transfer happens by copying or by transformation.  However, as we have argued above, the difference between copying and transformation is essential in many biological scenarios, as well as other domains.  In such cases,  $\Dcopy$ arises naturally as the minimal thermodynamic work required to replicate information  by copying.

Concerning the example of DNA copying that we discussed  throughout this section, our results should be interpreted with some care. 
We have generally imagined that the source system represents the state of an entire polymer, e.g., the state of an entire DNA molecule,  and that the probability of exact copying refers to the probability that the entire sequence  is reproduced without any errors.  
Alternatively, one can use the same framework to consider probability of copying a single monomer in a long polymer (assuming that the thermodynamics of polymerization can be disregarded), as might be represented for instance by a single-nucleotide DNA substitution matrix~\cite{yang1994estimating}, as analyzed in the last section. Generally speaking,  $\Dcopy$ computed at the level of single monomers will be different from $\Dcopy$ computed at the level of entire polymers, since the probability of exact copying means different things in these two formulations. 






\section{Copy and transformation in amino acid substitution matrices}
\label{sec:PAM}

\begin{figure}
\centering
\includegraphics[width=1\linewidth]{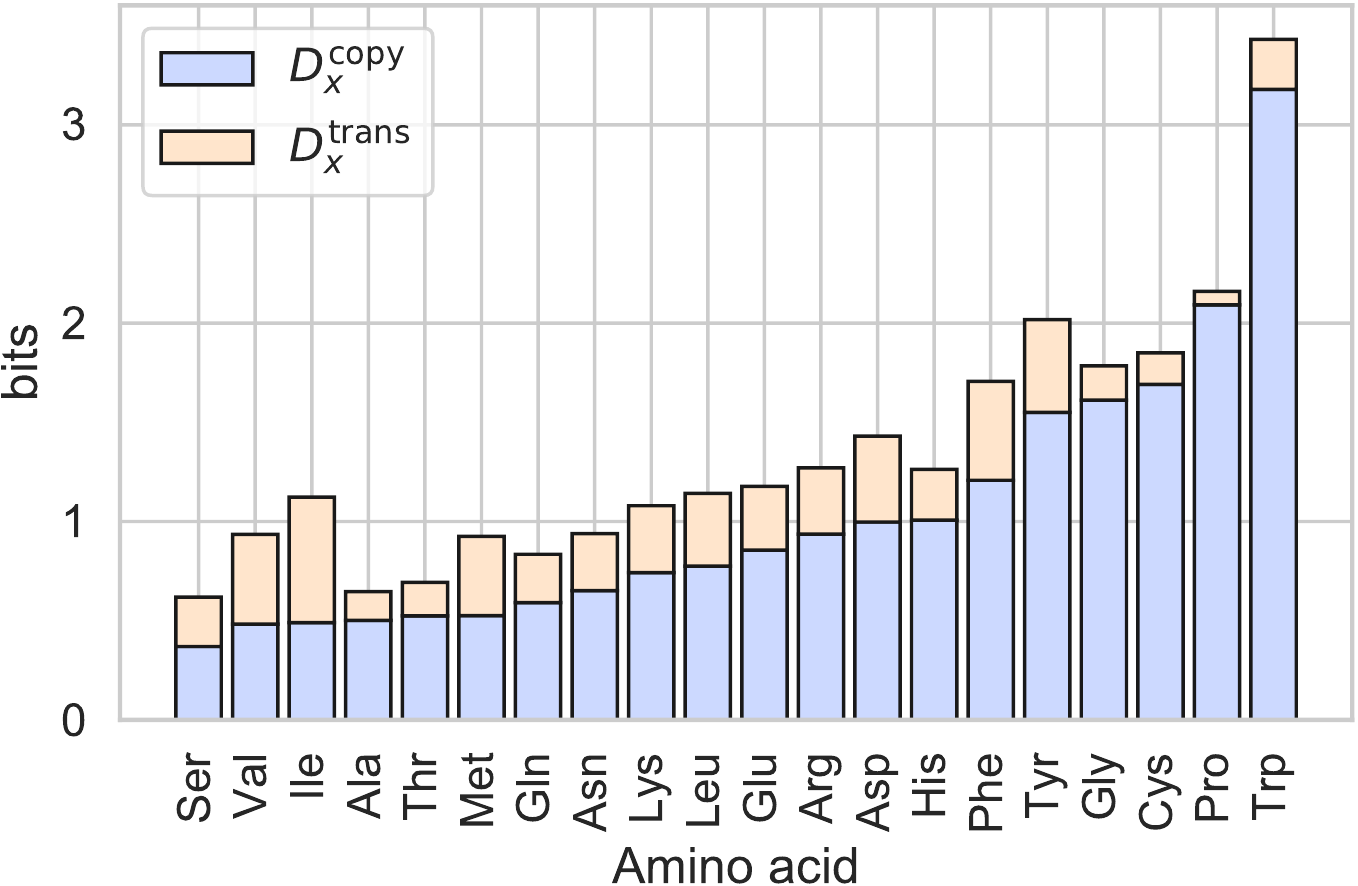}
\caption{
Copy and transformation information for different amino acids, based on an empirical PAM matrix~\cite{le2008improved}. We show magnitude of $\DcopyPP$ in blue; on top of this is the amount of transformation information in orange. The sum of both is the specific MI for each amino acid, according to the decomposition given in \cref{eq:Ip2}.
}
\label{fig:aa}
\end{figure}

In the previous section, 
we saw how $\Dcopy$ and $\Icopy$ arise naturally when studying the fundamental limits on the thermodynamics of copying, which includes the special case of replicating information-bearing polymers.  
Here we demonstrate how these measures can be used to characterize the information-transmission properties of a real-world biological replication system, as formalized by a communication channel $p_{Y|X}$ from parent to offspring~\cite{yang1994estimating,le2008improved}. 
%
In this context, we show how $\Icopy$ can be used to quantify precisely how much  
information is transmitted by copying, without  mutations. 
At the same time, we will use $\Itransform$ to quantify how much information is transmitted by transformation, that is by systematic \emph{nonrandom} mutations that carry information but do not preserve the identity of the original message~\cite{li1984nonrandomness,yang1994estimating,graur_fundamentals_2000}. We also quantify the effect of purely-random mutations, which correspond to the conditional entropy of the channel, $H({Y|X})$.  


We demonstrate these measures on empirical data of \emph{point accepted mutations} (PAM) of amino acids.  PAM data represents 
the rates of substitutions between different amino acids during the course of biological evolution, and has various applications, including evolutionary modeling, phylogenetic reconstructions, and protein alignment~\cite{le2008improved}.  We emphasize that amino acid PAM matrices do not reflect the direct physical transfer of information from protein to protein, but rather the effects of underlying processes of DNA-based replication and selection, followed by translation.

Formally, an amino acid PAM matrix $Q$ is a continuous-time rate matrix. $Q_{yx}$ represents the instantaneous rates of substitutions from amino acid $x$ to amino acid $y$, where  both $x$ and $y$ belong to $\Alphabet = \{1,\dots,20\}$, representing the 20 standard amino acids.  
We performed our analysis on a particular PAM matrix $Q$ which was published by Le and Gascuel~\cite{le2008improved} (this matrix was provided by the \texttt{pyvolve} Python package~\cite{spielman2015pyvolve}). 
We calculated a discrete-time conditional probability distribution $p_{Y|X}$ from this matrix by computing the matrix exponential $p_{Y|X} = \exp(\tau Q)$. Thus, $p(y|x)$ represents the probability that amino acid $x$ is replaced by amino acid $y$ over time scale $\tau$.  For simplicity we used timescale $\tau=1$. 
We used the stationary distribution of $Q$ as the source distribution $s_X$, which correlates closely with empirically-observed amino acid frequencies~\cite[Fig.~1]{le2008improved}.  
Using the decomposition presented in \cref{eq:Itransformtotal}, we arrived at the following values for the communication channel described by the conditional probabilities $p_{Y|X}$:
 \[
 I\YX = \Icopy\XrY + \Itransform\XrY\approx 1.2\;{\rm bits} \eqeol,
 \]
 where 
 \begin{align}
 \Icopy\XrY& = \sum_x s(x) \DcopyPP \approx 0.88 \; {\rm bits}\eqeol, 
 \nonumber\\
 \Itransform\XrY & = \sum_x s(x) \DtransPP \approx 0.32\;{\rm bits}\eqeol. \nonumber
 \end{align}
We also computed the intrinsic noise for this channel (see \cref{eq:threeway}),
\[
H({Y|X}) = \sum_x s(x) H(Y|X=x) \approx 2.97 \;{\rm bits}\eqeol.
\]
Finally, we computed the specific copy and transformation information, $\Dcopy$ and $\Dtrans$, for different amino acids. The results  are shown in \cref{fig:aa}. We remind the reader that the sum of $\DcopyPP$ and $\DtransPP$ for each amino acid $x$ --- that is, the total height of the stacked bar plots in the figure --- is equal to the specific MI $I\Yx$ for that $x$, as explained in the decomposition of \cref{eq:Ip2}.

While we do not dive deeply in the biological significant of these results, we highlight several interesting findings.  First, for this PAM matrix and timescale ($\tau=1$), a considerable fraction of the information ($\approx 1/4$) is transmitted not by copying but by non-random mutations.  Generally, such non-random mutations  represent underlying physical, genetic, and biological constraints that allow some pairs of amino acids to substitute each other more readily than other pairs. 

Second, we observe considerable variation in the amount of specific MI, copy information, and transformation between different amino acids, as well as different ratios of copy information to transformation information.  In general,  amino acids with more copy information are conserved unchanged over evolutionary timescales. At the same time, it is known that conserved amino acids tend to be ``outliers'' in terms of their physiochemical properties (such as hydrophobicity, volume, polarity, etc.), since mutations to such outliers are likely to alter protein function in deleterious ways~\cite{graur1985amino,yang1998models}. 
To analyze this quantitatively, we used Miyata's measure of distance between amino acids, which is based on differences in volume and polarity~\cite{miyata1979two}. For each amino acid, we quantified its degree of ``outlierness'' in terms of its mean Miyata distance to all 19 other amino acids.  The Spearman rank correlation between this outlierness measure and copy information (as shown in \cref{fig:aa}) was 0.57 ($p=0.009$).  On the other hand, the rank correlation between outlierness and transformation information was 0.22 ($p=0.352$).  Similar results were observed for other chemically-motivated measures of amino acid distance, such as Grantham's distance~\cite{grantham1974amino} and Sneath's index~\cite{sneath1966relations}. This demonstrates that amino acids with unique chemical characteristics tend to have more copy information, but not more transformation information.

\section{Discussion}
\label{sec:discussion}

Although mutual information is a very common and successful measure of transmitted information, it is  insensitive to the  distinction between information that is transmitted by copying versus information that is transmitted by transformation. 
Nonetheless, as we have argued, this distinction  is of fundamental importance in many real-world systems. 

In this paper we propose a rigorous and practical way to decompose specific mutual information, and more generally Bayesian surprise, into two non-negative terms corresponding to copy and transformation,  $I=\Icopy+\Itransform$. 
We derive our decomposition using an axiomatic framework: we propose a set of four axioms that any measure of copy information should obey, and then identify the unique measure that satisfies those axioms.  
At the same time, we show that our measure  of copy information is one of a family of functionals, each of which corresponds to a different way of quantifying error in transmission. 
We also demonstrate that our measures have a natural interpretation in thermodynamic terms, which 
suggests novel approaches for understanding the thermodynamic efficiency of biological replication processes, 
 in particular DNA and RNA duplication.  Finally, we demonstrate our results on real-world biological data, exploring  copy and transformation information of amino acid substitution rates. We find significant variation among the amount of information transmitted by copying vs. transformation among different amino acids.


Several directions for future work present themselves. 

First, there is a large range  of  practical and theoretical application of our measures, from analysis of biological and neural information transmission to the study of the thermodynamics of self-replication, a fundamental and challenging problem in biophysics~\cite{Corominas-Murtra:2019}. 

Second, we suspect our measures of copy and transformation information 
have further connections to existing formal treatments in information theory, in particular 
rate-distortion theory~\cite{cover_elements_2012}, whose connections we started to explore here. We also believe that our decomposition may be generalizable beyond Bayesian surprise and mutual information to include other information-theoretic measures, including conditional mutual information and multi-information.  Decomposing conditional mutual information is of particular interest, since it will permit a decomposition of the commonly-used \emph{transfer entropy}~\cite{Schreiber:2000} measure into  copy and transformation components, thus separating two different modes of dynamical information flow between systems.

Finally, we point out that our proposed decomposition 
has some high-level similarities 
to other recent proposals for information-theoretic decomposition, such as the  ``partial information decomposition''  of multivariate information into redundant and synergistic components~\cite{williams2010nonnegative}, integrated information decompositions~\cite{kahle2009complexity,oizumi2016unified}, and decompositions of mutual information into ``semantic'' (valuable) and ``non-semantic'' (non-valuable) information~\cite{kolchinsky2018semantic}.  We also mention another recent proposal for an alternative information-theoretic notion of ``copying''~\cite{mediano2019beyond}, in which copying is said to occur in a multivariate system when information that is present in one variable spreads to other variables (regardless of any transformations that information may undergo).  Further research should explore if and how the decomposition proposed in this paper  relates to these other approaches.



\vspace{10pt}

\section*{Acknowledgments}
AK was supported by Grant No.
FQXi-RFP-1622 from the FQXi foundation, and Grant No. CHE-1648973 from the U.S. National Science Foundation. AK would like to thank the Santa Fe Institute for supporting this research. 
The authors thank Jordi Fortuny, Rudolf Hanel, Joshua Garland, and Blai Vidiella for helpful discussions, as well as the anonymous reviewers for their insightful suggestions.
 
\bibliography{ref}

\ifappendix
\appendix
\renewcommand{\thesection}{\Alph{section}}

\counterwithin{thm}{section}
\counterwithin{lem}{section}
\counterwithin{prop}{section}

\section{$\Dcopy$ satisfies the four axioms}
\label{app:satisfies}
$\Dcopy$ 
satisfies \cref{axiom:Nonnegative} by non-negativity of KL.

It satisfies \cref{axiom:UpperBound} when $\postx > \priorx$ because $\bKL(\postx, \priorx) \le \DPP$ by the data processing inequality for KL divergence~\cite[Lemma 3.11]{csiszar_information_2011}. Otherwise, when $\postx \le \priorx$, $\Dcopy$ vanishes and thus satisfies  \cref{axiom:UpperBound} trivially.  

It satisfies \cref{axiom:Monotonicity} when $\postx \le \priorx$ because in that case $ \Dcopy(\post \Vert \prior) = 0 \le \Dcopy(\postb \Vert \prior)$. If $\postx \le \priorx$, then note that the derivative of $\bKL(a,b)$ with respect to $a$ is
$\frac{d}{da}\bKL(a,b) = \log \frac{a}{b} - \log \frac{1-a}{1-b}$, 
which is strictly positive when $a>b$. Thus, $ \Dcopy(\post \Vert \prior) \le \Dcopy(\postb \Vert \prior)$.

Finally, we show that $\Dcopy$ satisfies \cref{axiom:purelycorrect}.
For any prior distribution $\prior$, define the following posterior distribution $\post^\alpha(y)$:
\begin{align}
\post^\alpha(y) = \begin{cases}
\alpha & \text{if $y=x$}\\
\frac{1-\alpha}{1-\priorx}\prior(y) & \text{if $y\ne x$}
\end{cases} \eqeol,
\label{eq:alphapost}
\end{align}
where $\alpha$ is a parameter that can vary from $\priorx$ to 1.  It is easy to verify that for all $\alpha$,
\begin{align}
\label{eq:alpha2}
\DKL(\post^\alpha \Vert \prior) = \bKL(\alpha, \priorx) = \Dcopy(\post^\alpha \Vert \prior)\eqeol,
\end{align}
and that $ \Dcopy(\post^\alpha \Vert \prior)$ ranges in a continuous manner from 0 (for $\alpha = \priorx$) to $-\log \priorx$ (for $\alpha=1$).



\section{Proof of Theorem~\ref{thm:unique}}
\label{app:proofs}

Before proceeding, we first prove two useful lemmas.

\begin{restatable}{lem}{lemone}
Given \cref{axiom:Monotonicity}, $\RR(\post, \prior,x)=\RR(\postb, \prior,x)$ if $\postx=\postbx$. 
\label{lem:R-function-of-Pe}
\end{restatable}
\begin{proof} 
Follows from applying \cref{axiom:Monotonicity} in both directions.
\end{proof}

\begin{restatable}{lem}{lemweak}
Given \cref{axiom:Monotonicity,axiom:UpperBound,axiom:Nonnegative}, if $\postx \le \priorx$, then $\RR(\post,\prior,x)=0$.
\label{lem:weakchannel}
\end{restatable}
\begin{proof}
If $\postx \le \priorx$, then $F(\post, \prior,x) \le F(\prior, \prior,x)$ by \cref{axiom:Monotonicity}.  By \cref{axiom:UpperBound}, $F(\prior,\prior,x) \le \DKL(\prior \Vert \prior)=0$.  Combining gives $F(\post, \prior,x) \le 0$, while $F(\post, \prior,x)\ge 0$ by \cref{axiom:Nonnegative}.
\end{proof}

We then show that $\Dcopy$ is the largest  possible measure that satisfies \cref{axiom:Nonnegative,axiom:UpperBound,axiom:Monotonicity}.

\begin{restatable}{prop}{propone}
Any  $F$ which satisfies \cref{axiom:Nonnegative,axiom:UpperBound,axiom:Monotonicity} must obey $F(\post, \prior,x) \le \DcopyPP$.
\label{prop:upperbound}
\end{restatable}
\begin{proof} 
Given \cref{lem:weakchannel}, without loss of generality we restrict our attention to the case where $\postx > \priorx$. 
Define the posterior $\post^\alpha$ as in \cref{eq:alphapost}, while taking $\alpha=\postx$. 
Then, by \cref{lem:R-function-of-Pe},
$$F(\post, \prior,x)  = F(\post^\alpha, \prior,x) \eqeol.$$
At the same time, 
\begin{multline*}
F(\post^\alpha, \prior,x) \le \DKL(\post^\alpha \Vert \prior)\\ = \bKL(\postx \Vert \priorx) = \DcopyPP,
\end{multline*}
where the first inequality follows from \cref{axiom:UpperBound}, and the second equality from \cref{eq:alpha2}.
\end{proof}

We are now ready to prove the main result from \cref{sec:AxiomsDI}.

\begin{proof}[Proof of \cref{thm:unique}]
Consider some $\post, \prior, x$, and assume $\postx > \priorx$ (without loss of generality by \cref{lem:weakchannel}). By \cref{axiom:purelycorrect}, there must exist a posterior $\postb$ such that $\postbx = \postx$ and 
\begin{align}
F(\postb, \prior, x) = \DKL(\postb \Vert \prior).
\label{eq:fd0}
\end{align}
Note that by the data processing inequality for KL divergence, $\DKL(\postb \Vert \prior) \ge  \Dcopy(\postb \Vert \prior)$.  

Then, by \cref{lem:R-function-of-Pe}, $F(\post, \prior, x)=F(\postb, \prior, x)$ since $\postx = \postbx$.  Similarly, it can be verified that $ \Dcopy(\postb \Vert \prior) =  \DcopyPP$.  Combining the above results shows that $F(\post, \prior, x) \ge \DcopyPP$. The theorem follows by combining with \cref{prop:upperbound}.

\end{proof}



\section{Axiomatic derivation and solution of Eq.~\ref{eq:genloss}}
\label{app:general}

\global\long\def\lxy{\ell(x,y)}%
\global\long\def\lxY{\ell(x,Y)}%
\global\long\def\dl{{\textstyle \frac{d}{d\lambda}}}%

\subsection{Axiomatic derivation}
\label{app:general-ax}

We first demonstrate that the generalized copy information defined in \cref{eq:genloss}, $\DcopyGenPP$, 
is the unique measure that satisfies \cref{axiom:Nonnegative,axiom:UpperBound} and our modified \cref{ax:mono2,ax:purely2}.  Our derivation  has the same structure as the one in \cref{app:proofs}, and we proceed more quickly. 

First, we verify that $\DcopyGen$ satisfies the four axioms.  It satisfies \cref{axiom:Nonnegative} by non-negativity of KL.  It satisfies \cref{axiom:UpperBound} because $\post$ falls within the feasibility set of \cref{eq:genloss}, therefore the minimum $\DcopyGenPP$ has to be less than or equal to $\DPP$.  It satisfies \cref{ax:mono2} because $\mathbb{E}_{\post}[\ell(x,Y)] \ge \mathbb{E}_{\postb}[\ell(x,Y)]$ means that the feasibility set of \cref{eq:genloss} for $\postb$ is a subset of the feasibility set for $\post$, so the minimum $\DcopyGenPPb$ has to be greater than or equal to the minimum $\DcopyGenPP$. To show that it satisfies \cref{ax:purely2}, note that the distribution $w_Y$ which optimizes \cref{eq:genloss} will achieve $\mathbb{E}_{w_Y}[\ell(x,Y)] = \mathbb{E}_{\post}[\ell(x,Y)]$ whenever $\mathbb{E}_{\post}[\ell(x,Y)] \le \mathbb{E}_{\prior}[\ell(x,Y)]$~\cite[pp.299-300]{rubinstein_simulation_2016}. Note also that $ \mathbb{E}_{\post}[\ell(x,Y)]$ can vary from $\min_y \ell(x,y)$ (for $\post(y|x)=\delta(y,\argmin_{y'} \ell(x,y')$) to $\mathbb{E}_{\prior}[\ell(x,Y)]$ (for $\post = \prior$).

We now demonstrate that $\DcopyGen$ is the unique measure that satisfies the four axioms. 
We begin by showing that $\RR(\post, \prior, x) \le \DcopyGenPP$ for any $\RR$.  Given a choice of $\post$, $\prior$, and $x$, let $w_Y$ be the solution to \cref{eq:genloss}, so
\begin{align}
\label{eq:subbf}
\DcopyGenPP = \DKL(w_Y \Vert \prior) .
\end{align}
Given the definition of $\DcopyGen$, 
$\mathbb{E}_{w_Y}[\ell(x,Y)] \le \mathbb{E}_{\post}[\ell(x,Y)]$. Then, by \cref{ax:mono2}, \cref{axiom:UpperBound}, and \cref{eq:subbf}, 
\begin{multline*}
\RR(\post, \prior,x) \le \RR(w_Y, \prior,x) \\
\le  \DKL(w_Y \Vert \prior) = \DcopyGenPP.
\end{multline*}

We finish by showing that $\RR(\post, \prior, x) \ge \DcopyGenPP$ for any $\RR$. 
First consider the case $\mathbb{E}_{\post}[\ell(x,Y)] \ge \mathbb{E}_{\prior}[\ell(x,Y)]$. Then, $\DcopyGenPP = 0$ by construction, and therefore $\RR(\post, \prior, x) \ge \DcopyGenPP$ by \cref{axiom:Nonnegative}. 

When $\mathbb{E}_{\post}[\ell(x,Y)] < \mathbb{E}_{\prior}[\ell(x,Y)]$,  
by \cref{ax:purely2} there must exist a posterior $\postb$ such that $\mathbb{E}_{\postb}[\ell(x,Y)] = \mathbb{E}_{\post}[\ell(x,Y)] $ and 
\begin{align}
\RR(\postb, \prior, x) = \DKL(\postb \Vert \prior).
\label{eq:nfd0}
\end{align}
Then, by definition of $\DcopyGen$, 
\begin{align}
\DKL(\postb \Vert \prior) \ge \DcopyGenPP.
\label{eq:nfd1}
\end{align}
Finally, by \cref{ax:mono2},
\begin{align}
\RR(\post, \prior, x) & \ge \RR(\postb, \prior, x)\label{eq:nfd2}
\end{align}
Combining \cref{eq:nfd2}, \cref{eq:nfd0}, and then \cref{eq:nfd1} shows that $\RR(\post, \prior, x) \ge \DcopyGenPP$. 

Thus, $\DcopyGen$ is the unique measure that satisfies \cref{axiom:Nonnegative,axiom:UpperBound} and our generalized \cref{ax:mono2,ax:purely2}.

\subsection{$\Dcopy$ as the solution to Eq.~\ref{eq:genloss} for the 0-1 loss function}
\label{app:derivation}
Consider the optimization problem:
\begin{align}
\min_{{r_Y \in \SS:\,r_Y(x) \ge \postx }} \DKL(r_Y \Vert \prior)  \eqeol.
\label{eq:appopt}
\end{align}
When $\priorx \ge \postx$, then the solution $r_Y = \prior$ satisfies the constraint and achieves $\DKL(\prior \Vert \prior) = 0$, the minimum possible.  When $\priorx < \postx$, we use the chain rule for KL divergence~\cite{cover_elements_2012} to write
\begin{multline*}
\DKL(r_Y \Vert \prior)  = \bKL(r_Y(x), \priorx)  +\\
(1-r_Y(x)) \DKL(r_Y(Y\vert Y \ne x)\Vert \prior(Y\vert Y \ne x)) \eqeol.
\end{multline*} 
The second term is minimized by setting $r_Y(y) \propto \prior(y)$ for $y\ne x$, so that $r_Y(y\vert Y\ne x)=\prior(y\vert Y \ne x)$ and $\DKL(r_Y(Y\vert Y \ne x)\Vert \prior(Y\vert Y \ne x))=0$.  Thus, in the case that $\priorx < \postx$, we have reduced the optimization problem of \cref{eq:appopt} to the equivalent problem
\begin{align}
\min_{a \in [\postx,1]} \bKL(a, \priorx)  \eqeol.
\label{eq:appopt2}
\end{align}
Note that the derivative $\bKL(a,b)$ with respect to $a$ is
$\frac{d}{da}\bKL(a,b) = \log \frac{a}{b} - \log \frac{1-a}{1-b}$, 
which is strictly positive when $a>b$. 
Given the assumption that $\postx > \priorx$, \cref{eq:appopt2} is minimized by $a = \postx$.  Thus, $\bKL(\postx, \priorx)$ is the solution to \cref{eq:appopt} when $\priorx < \postx$.

Combining these two results shows that $\DcopyPP$, as defined in \cref{eq:DI2}, is the solution to \cref{eq:appopt}. 

\subsection{Vector-valued loss functions}
\label{app:general-vector}

One can also generalize the approach described in \cref{sec:general} to vector-valued loss functions, $\ell : \mathcal{X} \times \mathcal{Y} \to \mathbb{R}^n$, where we use $\mathcal{X}$ and  $\mathcal{Y}$ to indicate the sets of outcomes of $X$ and $Y$ respectively (recall that these can be different, in the context of our generalized copy and transformation information measures).  As we'll see below, one application of vector-valued loss functions is to define measures of copy and transformation information that are additive when independent channels are concatenated.

We first discuss which axioms might be expected to hold for generalized copy information measures with vector-valued loss functions.  \cref{axiom:Nonnegative} and \cref{axiom:UpperBound} do not make reference to the loss function, and remain unmodified. \cref{ax:mono2} is still meaningful, as long as the inequality $\mathbb{E}_{\post}[\ell(x,Y)] \ge \mathbb{E}_{\postb}[\ell(x,Y)]$ is taken in an element-wise fashion.  \cref{ax:purely2} should be dropped for vector-valued functions, for reasons explained below.

Using the derivation found in \cref{app:general-ax}, it can be shown that the largest measure which satisfies \cref{axiom:Nonnegative}, \cref{axiom:UpperBound}, and \cref{ax:mono2} for a vector-valued loss function  is given by
\begin{align}
\label{eq:genlossVV}
&\DcopyGenPP := \min_{r_Y}\;  \DKL(r_Y \Vert \prior)\\
&\qquad \text{s.t.} \quad 
\mathbb{E}_{r_Y}[\ell_i(x,Y)] \le \mathbb{E}_{\post}[\ell_i(x,Y)] \text{ for $i=1..n$}, \nonumber
\end{align}
where $\ell_i$ indicates the $i^\mathrm{th}$ component of the loss function $\ell$. \cref{eq:genlossVV} is a  minimum cross-entropy problem with $n$ different constraints.  The general solution to this problem will have the following form~\cite{rubinstein_simulation_2016}:
\begin{align}
w(y) =  \frac{1}{Z(\lambda_1,\dots,\lambda_n)} p_Y(y) e^{-\sum_i \lambda_i \ell_i(x,y)}, \label{eq:multsol}
\end{align}
where $\lambda_i \ge 0$ is the Lagrange multiplier for constraint $i$ and $Z(\lambda_1,\dots,\lambda_n)$ is a normalization constant.  The Lagrange multipliers can be found by using standard convex optimization techniques.  Note that all $\lambda_i=0$ if $\mathbb{E}_{\post}[\ell_i(x,Y)] \ge  \mathbb{E}_{\prior}[\ell_i(x,Y)]$ for all $i$, in which case $w_Y = p_Y$.  Even if $\mathbb{E}_{\post}[\ell(x,Y)] <  \mathbb{E}_{\prior}[\ell(x,Y)]$, however, it may be impossible to make all of the constraints simultaneously tight  up to equality. In other words, it will not always be the case that 
$\mathbb{E}_{w_Y}[\ell_i(x,Y)] = \mathbb{E}_{\post}[\ell_i(x,Y)]$ for all $i=1..n$, and some (but not necessarily all) of the multipliers $\lambda_i$ will be equal to 0.  For this reason, \cref{ax:purely2} is not generally achievable for copy information defined with vector-valued loss functions, and we drop it from our requirements. This means $\DcopyGen$, as defined in \cref{eq:genlossVV}, is not the unique measure which satisfies the remaining three axioms (\cref{axiom:Nonnegative}, \cref{axiom:UpperBound}, and \cref{ax:mono2}). For example, they are also satisfied by the trivial measure $F(\post,\prior,x) = 0$ for all $\post$, $\prior$, and $x$.

Vector-valued loss functions can be used to derive an additive measure of copy information.  Imagine that source and destination messages consists of sequences of $n$ symbols. If the source symbols are chosen independently, $s(x) = \prod_{i=1}^{n} s_i(x_i)$, and transmitted across $n$ independent channels, $p(y|x) = \prod_{i=1}^n p_i(y_i|x_i)$, then one can verify that the destination marginal distribution will also have a product form,
\begin{align}
\label{eq:appprod}
p(y) = \prod_{i=1}^n p_i(y_i).
\end{align}
In that case, one may desire a measure of copy information that is additive across the $n$ transmissions (see also discussion in \cref{sec:midecomp}).  This can be achieved by choosing an $n$-dimensional loss function, $\ell(x,y) = \langle  \ell_1(x_1,y_1), \ell_2(x_2,y_2), \dots, \ell_n(x_n,y_n)\rangle$.  It can be seen from \cref{eq:appprod} and \cref{eq:multsol} that the optimal distribution will have a product form, $w(y) = \prod_{i=1}^n w_i(y_i)$.  By  \cref{eq:genlossVV}, it can also be checked that the resulting copy information will have an additive form,
\begin{align}
\label{eq:genadditive}
\DcopyGenPP = \sum_{i=1}^n \DcopyGen(p_{Y_i|x_i} \Vert  p_{Y_i}) ,
\end{align}
where $\DcopyGen(p_{Y_i|x_i} \Vert  p_{Y_i})$ is the generalized copy information defined for  loss function $\ell_i(x_i, y_i)$.  Note that in this case $\DPP=\sum_i \DKL(p_{Y_i|x_i}\Vert p_Y)$. Therefore, by \cref{eq:dtransgen,eq:genadditive}, the generalized transformation information $\DtransGen$ will also be additive.

\section{Proof of Prop.~\ref{prop:dmi}} 
\label{app:purech}



Before proving \cref{prop:dmi}, we prove several intermediate results. 
We start by deriving some useful properties of the roots of the quadratic polynomial $ax^2 - (a+s)x+sc$. In particular, we consider the two roots
\begin{align}
f_\pm(a, s, c) = \frac{a+s \pm\sqrt{\left(a+s\right)^{2}-4asc}}{2a}
\label{eq:lem1c0}
\end{align}
where $a \in \mathbb{R} \setminus \{0\}$, $s \in (0,1]$, $c \in (0,1]$.

\begin{lem}
\label{lem:qbounds}
 $f_{+}(a, s, c) < 0$ when $a< 0$ and $f_+(a,s,c) \ge 1$ when $a > 0$.
 \end{lem}
 \begin{proof}
When $a<0$, $f_{+}(a, s, c)\le f_{-}(a, s, c)$.  
Vieta's formula states that
\begin{align}
f_{-}(a, s, c)  f_{+}(a, s, c)  & = \frac{sc}{a} < 0 \label{eq:lem1c1} \eqeol.
\end{align}
This implies $f_{+}(a, s, c)<0$.  When $a>0$, we lower bound the determinant,
\begin{align}
(a+s)^{2}-4asc \ge a^{2}+2as+{s}^{2}-4as=(a-s)^{2}\,.\label{eq:lem1c3}
\end{align}
This implies 
\begin{align*}
f_+(a, s, c) &\ge \frac{a+s + |a-s|}{2a} = \begin{cases}
1 & \text{if $a \ge s$} \\
\frac{s}{a} > 1 & \text{if $s > a > 0$} 
\end{cases}
\end{align*}
 \end{proof}

 \begin{lem}
 \label{lem:qlims}
 $\lim_{a\rightarrow 0} f_-(a,s,c) = c.$
 \end{lem}
 \begin{proof}
 By L'H\^{o}pital's rule,
\begin{align*}
\lim_{a\rightarrow0} f_{-}(a,s,c) &=\lim_{a\rightarrow0}\frac{\frac{d}{da}\left(a+s-\sqrt{\left(a+s\right)^{2}\!-\!4asc}\right)}{\frac{d}{da}( 2a ) }\\ 
& = \frac{1}{2}- \lim_{a\rightarrow0} \frac{2(a+s)-4sc}{2\cdot 2\sqrt{\left(a+s\right)^{2}-4asc}} \\
& =\frac{1}{2}-\frac{s-2sc}{2s} = c \eqeol.
\end{align*}
\end{proof}

\begin{lem}
\label{lem:qmono}
$f_-(a,s,c)$ is continuous and monotonically decreasing in $a$. It is strictly monotonically decreasing in $a$ when $f_-(a,s,c)<1$.
\end{lem}
\begin{proof}
First consider the the case when $c=1$,
\begin{align*}
f_-(a,s,c) & = 
\frac{a+ s- |a-s|}{2a} = \begin{cases}
\frac{s}{a} & \text{if $a\ge s$} \\
1 & \text{otherwise}
\end{cases} 
\end{align*}
which is continuous and monotonically decreasing in $a$, and strictly so when $f_-(a,s,c)<1$ (so $a > s$). 

When $c < 1$, define the square root of the determinant
\[
\eta := \sqrt{\left(a+s\right)^{2}-4asc} \stackrel{(a)}{>} |a -s| \ge 0\eqeol.
\]
Inequality $(a)$ is strict because \cref{eq:lem1c3} is strict when $c<1$. 
Then, consider the derivative,
\begingroup
\allowdisplaybreaks
\begin{align}
& \textstyle{\frac{\partial}{\partial a}} f_-(a,s,c) \nonumber \\
&=\frac{1}{4a^{2}}\Big[\left(1-\frac{1}{2}\frac{2a+2s-4sc}{\eta}\right)2a - 2\left(a+s-{\eta}\right)\Big] \nonumber \\
 &=\frac{1}{2a^{2}}\left[-\frac{a^2+sa-2sc a}{\eta}-s+{\eta}\right] \nonumber \\
 & \propto-a^{2}-sa+2sca-s{\eta}+{\eta}^{2} \label{eq:rel1}\\
 & = s\left[ a-2ac+s-{\eta} \right]   \label{eq:rel2}\\
 & \propto\frac{a-2ac+s}{{\eta}}-1  \label{eq:rel3}\\
 & \le \frac{\left| a-2ac+s \right|}{{\eta}} - 1 \nonumber \\
 & = \sqrt{\frac{\left(a-2ac+s\right)^{2}}{{\eta}^{2}}}-1 \nonumber \\
 & =\sqrt{1-4a^{2}c\frac{1-c}{{\eta}^{2}}}-1 < 0\eqeol, \nonumber 
\end{align}
\endgroup
where in  \cref{eq:rel1} we multiplied by the (positive) term ${2a^{2}{\eta}}$,
in \cref{eq:rel2} we plugged in the definition of $\eta$ and simplified,
 and in \cref{eq:rel3} we divided by the (strictly positive) term ${\eta}s$. 
 The inequality
in the last line uses the fact that $4a^{2}c\frac{1-c}{{\eta}^{2}} > 0$ given that $a\ne 0$ and $0<c<1$, and that $\sqrt{1-x}<1$ for $x>0$.
\end{proof}


We now prove the following.
\def\astar{a^*}
\def\pax{\py^{\astar}(x)}

\begin{thm}
\label{prop:aux}
Let $\ax \in [0,1]$ indicate a set of values for all $x \in \Alphabet$. 
Then, for any source distribution $s_X$ with full support, there is a channel $p_{Y|X}$ that satisfies
\begin{equation}
\pyx =\begin{cases}
\ax & \text{if}\;x=y\\
\frac{1-\ax}{1-\py(x)}\py(y) & \text{otherwise}\eqeol,
\end{cases}\label{eq:soln-prop1A}
\end{equation}
where $\py$ is the marginal $\py(y)=\sum_{x}s(x)\pyx$. The channel $p_{Y|X}$ is unique if $\ax > 0$ for all $x$. 
Moreover, $I_p\YX=\Icopy_p\XrY$ if and only if  $\sum_x \ax \ge 1$.
\end{thm}
\begin{proof}
We will show that there exists a marginal  $\py$ that satisfies the consistency conditions of \cref{eq:soln-prop1A}.  

We first eliminate a few edge cases. The solution is trivial for $|\Alphabet|=1$, so we assume that $|\Alphabet| \ge 1$.  
If  $\ax = 0$ for all $x$, then for any two states $x,x' \in \Alphabet$, the following is a solution: 
$p_Y(x) = s(x')/(s(x)+s(x'))$, $p_Y(x') = s(x)/(s(x)+s(x'))$, $p_Y(x'')=0$ for all $x'' \in \Alphabet \setminus \{x,x'\}$. 
If $\ax = 0$ for some but not all $x$, then the problem can be solved for the reduced outcome space $\mathcal{S}=\{ x \in \Alphabet : \ax > 0\}$, using  the procedure below. It can then be extended to all outcomes by keeping $p_Y(x)$ fixed for $x\in \mathcal{S}$ and setting $p_Y(x) = 0$ for all $x \in \Alphabet \setminus \mathcal{S}$. 
Therefore, without loss of generality, below we assume $\ax >0$ for all $x$. 





We now plug \cref{eq:soln-prop1A} into $\py(y)=\sum_{x}s(x)\pyx$, 
\begin{align}
\py(x) & =s(x)\ax +\py(x)\sum_{x':x'\ne x}s(x')\frac{1-\axx}{1-\py(x')} \eqeol.
\label{eq:lll0}
\end{align}
Define $a:=1 -\sum_{x'}s(x')\frac{1-\axx}{1-\py(x')}$ and rearrange \cref{eq:lll0} to give
\begin{align*}
0 & =s(x)\ax +\py(x)\left(-a - s(x)\frac{1-\ax}{1-\py(x)} \right) \eqeol .
\end{align*}
Multiplying both sides by $1-\py(x)$ and simplifying gives
\begin{align}
0 & =s(x)\ax-s(x)\ax \py(x)-a\py(x)+a{\py(x)}^{2}- \nonumber \\
 & \qquad\left[\py(x)s(x)-\py(x)s(x)\ax\right] \nonumber \\
 & =a{\py(x)}^{2}-(a+s(x))\py(x)+s(x)\ax \label{eq:lll4} 
 \eqeol.
\end{align}
Dividing by $s(x)$, then summing over $x$ and rearranging gives
\begin{align}
a\left[\sum_x \frac{{\py(x)-\py(x)}^{2}}{s(x)}\right]  =  \left[ \sum_x \ax\right] - 1 \,.
\label{eq:lll6}
\end{align}
Note that the sum inside the brackets on the left hand side is strictly positive. Thus,we have
\begin{align}
a \ge 0  \text{ iff } \sum_x \ax \ge 1 \quad ; \quad a < 0  \text{ iff } \sum_x \ax < 1
\label{eq:mmm0}
\end{align}
Note also that $a=0$ if $\sum_x \ax = 1$, in which case  $\py(x)=c(x)$ is the unique solution to \cref{eq:lll4} for all $x$. Below, we disregard this special case, and assume that $\sum_x \ax \ne  1$ and $a\ne 0$. 


We now solve \cref{eq:lll4} for $\py(x)$. 
First, note that  $\py(x)=\sum_{x'}s(x')p(x|x') \ge s(x) c(x) > 0$ for all $x$, since we assume that $s(x)>0$ and $c(x)>0$ for all $x$. Given that $|\Alphabet| > 1$, this also means that $\py(x) < 1$ for all $x$ (if this were not the case, then it would have to be that $\py(x) =0$ for all except one $x$). We then solve the quadratic equation,
\begin{equation}
\py^{a}(x)=\frac{a+s(x) -  \sqrt{\left(a+s(x)\right)^{2}-4a s(x) \ax}}{2a}\,,\label{eq:quadratic-1}
\end{equation}
where we include the superscript $a$ in $\py^a$ to make the dependence on
$a$ explicit. We chose the negative solution of the quadratic equation because, by \cref{lem:qbounds}, it is the only one compatible with the requirement that $0 < \py^a(x) < 1$.

We wish to find the value of $a$ satisfies $\sum_{x}\py^a(x)=1$, which 
is defined implicitly via 
\begin{align}
1 & =\sum_x \frac{a +s(x) -  \sqrt{\left(a +s(x)\right)^{2}-4 a s(x) \ax}}{2a } \label{eq:solutiona}
\end{align}
Note that  each $\py^{a}(x)$ is continuous and strictly monotonically decreasing in $a$ (\cref{lem:qmono}), and therefore so is the right hand side of \cref{eq:solutiona}. Moreover, $a$ must lie between $-1$ and $1$. To see why, evaluate the right hand side of \cref{eq:solutiona} for $a=-1$,
\begin{align*}
&	\sum_{x}\frac{1-s(x)+\sqrt{(1+s(x))^{2}+4s(x)c(x)}}{2}\\
&\ge \sum_{x}\frac{1-s(x)+(1+s(x))}{2}=n\ge1
\end{align*}
Then, evaluate it for $a=1$,
\begin{align*}
&	\sum_{x}\frac{1+s(x)-\sqrt{(1+s(x))^{2}-4s(x)c(x)}}{2}\\
& \le \sum_{x}\frac{1+s(x)-\sqrt{(1+s(x))^{2}-4s(x)}}{2}\\
& = \sum_{x}\frac{1+s(x)-(1-s(x))}{2}=\sum_{x}\frac{2s(x)}{2}=1 
\end{align*}
Thus, there is a unique $a \in [-1,1]$ that satisfies \cref{eq:solutiona}, resulting in a unique $\py^{a}$ and corresponding $p_{Y|X}$ in \cref{eq:soln-prop1A}.

Now, by definition of $\Icopy$ and the channel $p_{Y|X}$ in \cref{eq:soln-prop1A}, $I_p\YX=\Icopy_p\XrY$ if $c(x) \ge \py(x)$ for all $x$.  By \cref{lem:qlims} and \cref{lem:qmono}, the right hand side of \cref{eq:quadratic-1} is greater than $\ax$ if and only if $a \ge 0$. By \cref{eq:mmm0}, $a \ge 0$ if and only if $\sum_x \ax \ge 1$.


\end{proof}

In practice, the value of $a$ which satisfies \cref{eq:solutiona} in the proof of \cref{prop:aux} can be found by a numerical root finding algorithm, or by trying values from $-1$ to $1$ in small
intervals and selecting the first value that makes the LHS of \cref{eq:solutiona}
less than or equal to $1$. The marginal $\py$ and channel $p_{Y|X}$ can then be computed in closed
form using \cref{eq:soln-prop1A,eq:quadratic-1}.

We are now ready to prove \cref{prop:dmi}.

\begin{manualprop}{\ref{prop:dmi}}
For any source distribution $s_X$ with $H(X) < \infty$, there exist channels $p$ for all levels of  mutual information $I_p\YX \in[0,H(X)]$ such that $\Icopy_p\XrY = I_p\YX$.
\end{manualprop}
\begin{proof} 
Consider the proof of \cref{prop:aux}.  Note that for each $x \in \Alphabet$ and any $\gamma \in [0, 1]$, \cref{eq:lll4}  is satisfied by taking $\py(x) = s(x)$ and $c(x)= \gamma + s(x) - \gamma s(x)$.  

Let $p_{Y|X}^\gamma$ represent the channel corresponding to each $\gamma$, as defined in \cref{eq:soln-prop1A}.  
It is easy to check that $\Icopy_{p^\gamma}\XrY = I_{p^\gamma}\YX$, with $\Icopy_{p^\gamma}\XrY = 0$ for $\gamma=0$ and $\Icopy_{p^\gamma}\XrY = H(s_X)$ for $\gamma=1$.  Note that $c(x)$ is increases monotonically in $\gamma$ for all $x$, from $c(x) = s(x)$ for $\gamma=0$ to $c(x)=1$ for $\gamma=1$. This means that for all $\gamma$, 
\begin{align*}
\Icopy_{p^\gamma}\XrY &= \sum_x s(x) \bKL(c(x), s(x)) \\
& \le \sum_x s(x) \bKL(1, s(x))\\
& = -\sum_x s(x) \ln s(x) = H(s_X) < \infty .
\end{align*}
Thus, the sums that define $\Icopy_{p^\gamma}\XrY$ for each $\gamma$ converge uniformly, so $\Icopy_{p^\gamma}\XrY$ is  continuous in $\gamma$.   The proposition follows from the intermediate value theorem.
\end{proof}




\section{The binary symmetric channel}
\label{app:bsc}

The BSC is a channel over a two-state space ($\Alphabet = \{0,1\}$)  parameterized by a ``probability of error'' $\epsilon\in[0,1]$. The BSC can be represented in matrix form as 
\[
p_{{Y|X}}^\epsilon=\left(
\begin{array}{ll}
1-\epsilon & \epsilon  \\ 
\epsilon & 1-\epsilon
\end{array}
\right) \eqeol.
\]
When $\epsilon=0$, the BSC is a noiseless channel which copies the source without error. In this extreme case, MI is large, and we expect it to consist entirely of copy information. 
On the other hand, when $\epsilon=1$, the BSC is a noiseless ``inverted'' channel, where messages are perfectly switched between the source and the destination.
In this case, MI is again large, but we now expect it to consist entirely of transformation information.  
Finally, 
$\epsilon=1/2$ defines a completely noisy channel, for which mutual information (and thus copy and transformation information) must be 0.

For simplicity, we assume a uniform source distribution, $s_X(0)=s_X(1)=1/2$, which by symmetry implies a uniform marginal probability $p_Y(0)=p_Y(1)=1/2$ at the destination for any $\epsilon$.  
For the BSC with this source distribution,  \cref{eq:DI2} states that for both $x=0$ and $x=1$,
$\DcopyPPe = I_{p^\epsilon}\Yx$ and $\DtransPPe = 0$ when $\epsilon \le 1/2$, and $\DcopyPPe = 0$ and $\DtransPPe = I_{p^\epsilon}\Yx$ otherwise.
Using the definition of the (total) copy and transformation components of total MI, \cref{eq:Icopytotal,eq:Itransformtotal}, it then follows that 
\begin{align*}
\Icopy_{p^\epsilon}\XrY& = \begin{cases}
I_{p^\epsilon}\YX  & \text{if $\epsilon \le 1/2$} \\
0  & \text{otherwise}
\end{cases}\\
\Itransform_{p^\epsilon}\XrY& = \begin{cases}
I_{p^\epsilon}\YX  & \text{if $\epsilon \ge 1/2$} \\
0  & \text{otherwise}
\end{cases}\eqeol.
\end{align*}
This confirms intuitions about the BSC discussed in the beginning of this section.  The behavior of MI, $\Icopy\XrY$ and $\Itransform\XrY$ for the BSC with a uniform source distribution is shown visually  in \cref{Fig:Binary_Ch} of the main text.

\clearpage

\fi




\end{document}